\documentclass[10pt,a4paper]{article}

\usepackage[T1]{fontenc}
\usepackage[utf8]{inputenc}
\usepackage{mathrsfs}       
\usepackage{calligra}      
\usepackage{verbatim}      

\usepackage{amsmath, amsfonts, amssymb, amscd, amsthm}
\usepackage[all]{xy}       

\usepackage[margin=2.5cm]{geometry} 

\usepackage{graphicx}
\usepackage{tikz}
\usepackage{float}          
\usepackage[figuresright]{rotating} 

\usepackage{longtable}      
\usepackage{multirow}       
\usepackage{makecell}       
\usepackage{diagbox}        
\usepackage{booktabs}       

\usepackage[perpage, symbol*]{footmisc} 
\usepackage[justification=centering]{caption} 
\usepackage[numbers, sort&compress]{natbib}    

\usepackage[dvipsnames]{xcolor} 
\usepackage{hyperref}
\hypersetup{
    colorlinks=true,
    linkcolor=Maroon,     
    citecolor=NavyBlue, 
    urlcolor=PineGreen,   
    bookmarksnumbered=true
}

\newtheorem{lemma}{\textbf{Lemma}}
\newtheorem{definition}{\textbf{Definition}}

\newtheorem{example}{\textbf{Example}}
\newtheorem{algorithm}{\textbf{Algorithm}}

\newcommand{\F}{\mathbb{F}}

\begin{document}

\title{\bf Extension of MDS Codes over Finite Fields}

\author{Puyin Wang and Jinquan Luo
\thanks{The authros are with School of Mathematics and Statistics \& Hubei Key Laboratory of Mathematical Sciences, Central China Normal University, Wuhan China 430079.
The authors are supported by National Natural Science Foundation of China (No.12441102) and the Project of International Cooperation and Exchanges NSFC (No.12411540221)
}
\thanks{E-mails: p.wang98@qq.com(P.Wang),  luojinquan@mail.ccnu.edu.cn(J.Luo)}}
\maketitle

{\bf Abstract}:
In this paper, we investigate extending of $[n,k,n-k+1]$ non-Reed--Solomon(non-RS) MDS codes to $[n+1,k,n-k+2]$ non-RS MDS codes. We develop a practical iterative algorithm for constructing longer MDS codes while preserving their non-RS structure. The proposed algorithm is inspired by the meet-in-the-middle algorithm in cryptanalysis. On average, it requires at most $O(kq^k)$ multiplications, compared with $O\left(\binom{n}{k-1} q^{k-1} k^3\right)$ multiplications for generic exhaustive search. Applying the proposed algorithm to several known families of MDS codes, including elliptic-curve, twisted Reed--Solomon, column-twisted Reed--Solomon, cyclic, and missing-monomial MDS codes over fields of order at most $32$, we obtain new MDS codes of increased lengths that remain inequivalent to Reed--Solomon codes. In particular, we construct $[9,4,6]$ codes over $\F_9$ and $\F_{13}$, each with Schur-square dimension $9$. Over $\F_{17}$, we obtain two inequivalent $[10,4,7]$ codes with Schur-square dimensions $9$ and $10$, respectively. The latter is the dual of an extended $[10,6,5]$ code, demonstrating that the method also applies when $k>n/2$. For these dimension-four codes, the Schur-square dimensions $9$ and $10$ exceed both  that of  RS code (Schur-square dimension $7$) and those of specific elliptic-curve and column-twisted Reed--Solomon families (Schur-square dimension $8$). Some of these parameters of non-RS MDS codes have not been reported before.


\section{Introduction}

\qquad In coding theory, linear codes are very important to ensure reliable communication over noisy channels. Let $\F_q$ denote a finite field of size $q$. A linear code $C$ over $\F_q$ is simply a $k$-dimensional subspace of $\F_q^n$. We often use the parameters $[n, k, d]$ to describe it, where $d$ is the minimum distance(minimum weight) defined by
\[
d = \min\{\operatorname{wt}_H(\boldsymbol{c}) : \boldsymbol{c} \in C,\, \boldsymbol{c} \neq 0\}.
\]
Here, $\operatorname{wt}_H(\boldsymbol{c})$ is the Hamming weight, which counts the number of non-zero coordinates of $\boldsymbol{c}$.

It is well known that the parameters $n, k, d$ must satisfy the Singleton bound:
\[
d \leq n - k + 1.
\]
If a linear code reaches this equality, we call it a Maximum Distance Separable (MDS) code. Because of their optimal properties, MDS codes are widely used in practical applications, such as distributed storage systems and network coding \cite{CY}.

Due to their optimality with respect to the Singleton bound, MDS codes have been extensively studied. A classical family of MDS codes is Reed--Solomon (RS) codes, which are constructed by evaluating polynomials over finite fields \cite{RS}. However, there are also many other methods to construct MDS codes of non-RS type. Here, non-RS MDS refers to those MDS codes that are not equivalent to any RS code. These constructions can be roughly classified into the following categories:
\begin{itemize}
    \item Algebraic constructions: These include methods based on polynomial interpolation and algebraic geometry. In particular, algebraic geometry (AG) codes from curves provide a rich source of non-RS MDS codes \cite{CHHZ, WLLZ, HR}. Other algebraic approaches include cyclic constructions \cite{LCN} and some explicit families \cite{JMXZ, WLL, ZZ}.

    \item Twisted Reed--Solomon codes: These codes generalize RS codes by introducing twists in the evaluation structure. This method leads to many new non-RS MDS codes \cite{LLWZ, GZ, BPR, SYLH, SZS, ZL}.

    \item Extension and structural results: Another important direction focuses on the extension of MDS codes and the conditions to preserve the MDS property. Classical and recent results in this area include \cite{WDC}.
\end{itemize}

Constructing new MDS codes from known ones is an important topic. A direct way is to extend the code length while keeping the MDS property. However, whether a code can be extended is strictly limited by the field size and its algebraic structure. In fact, a general and easily verifiable criterion for extending MDS codes is still lacking, particularly for non-RS constructions.

To fill this gap, we provide a practical condition that allows us to systematically find longer non-RS MDS codes. Furthermore, we propose an efficient algorithm to significantly reduce the computational complexity required for such extensions. The main idea of our approach is to describe the admissible vectors as those avoiding a finite number of hyperplanes in $\F_q^k$.

The rest of this paper is organized as follows. Section~2 introduces some basic knowledge that will be used later. Section~3 presents our main results along with some examples. Finally, Section~4 concludes the paper and discusses future work.

\section{Preliminary}

\qquad In this section, we review some basic concepts and notations in coding theory. For more details, one can refer to \cite{MS}.

\begin{definition}
    Let $C$ be a $k$-dimensional linear code over $\F_q$ of length $n$. A \emph{generator matrix} of $C$ is a $k \times n$ matrix $G$ over $\F_q$ whose rows form a basis of $C$. Equivalently, $C$ can be expressed as
    \[
    C = \{ \boldsymbol{u}G \mid \boldsymbol{u} \in \F_q^k \}.
    \]
\end{definition}

Note that permutations of coordinates and multiplication on coordinate by nonzero scalars do not change the parameters of a code. Codes related by such operations are regarded as equivalent.

\begin{definition}
    Two linear codes $C$ and $C'$ of length $n$ over $\F_q$ are said to be equivalent if one can be obtained from the other by a combination of a permutation of coordinates and multiplication of coordinates by nonzero elements in $\F_q$.
\end{definition}

Let $\alpha_1,\dots,\alpha_n$ be distinct elements in $\F_q$ with $n \le q$. The $[n,k,n-k+1]$ Reed-Solomon (RS) code over $\F_q$ is defined by
\[
    \{c(f) = \big(f(\alpha_1), f(\alpha_2), \dots, f(\alpha_n)\big) \mid f(x) \in \F_q[x], \deg(f) < k \}.
\]
Moreover, RS codes are MDS codes. An MDS  code inequivalent to any RS code is called non-RS MDS code.

\begin{definition}

    Given two vectors $\boldsymbol{x}_1=(x_{1,1},\dots,x_{1,k})$ and $\boldsymbol{x}_2=(x_{2,1},\dots,x_{2,k})$ in $\F_q^k$, we say they are orthogonal if
    \[
    \langle \boldsymbol{x}_1, \boldsymbol{x}_2 \rangle = \sum_{i=1}^k x_{1,i}x_{2,i}=0.
    \]
\end{definition}

In this paper, extension of a code $C$ is the code obtained by appending one additional coordinate to each codeword of $C$, or equivalently, by adding one extra column to a generator matrix of $C$. We also refer to the resulting code as an extended code of $C$. Consequently, if the initial code $C$ is non-RS, the extended code of $C$ is also non-RS.

In order to determine whether a code is equivalent to an RS code, we need the following definition.

\begin{definition}
    Let $\boldsymbol{x}=(x_1,x_2,\dots,x_n)$ and $\boldsymbol{y}=(y_1,y_2,\dots,y_n)$ be two vectors in $\F_q^n$. The \emph{Schur product} (or componentwise product) of $\boldsymbol{x}$ and $\boldsymbol{y}$ is defined by
    \[
    \boldsymbol{x}*\boldsymbol{y} := (x_1 y_1, \dots, x_n y_n).
    \]
    For two linear codes $C_1, C_2 \subseteq \F_q^n$, their Schur product is defined as
    \[
    C_1 * C_2 := \mathrm{span}_{\F_q} ( \boldsymbol{x} * \boldsymbol{y} \mid \boldsymbol{x} \in C_1, \boldsymbol{y} \in C_2 ),
    \]
    where $\mathrm{span}_{\F_q}(S)$ denotes the $\F_q$-linear span of a set $S$.
\end{definition}

In particular, if $C_1 = C_2 = C$, we define $C^2 := C * C$ as the Schur square code of $C$. Clearly, $C_1^2$ and $C_2^2$ are equivalent when $C_1$ and $C_2$ are equivalent.

While the Schur square product code provides a convenient sufficient condition for non-RS equivalence, it may occasionally yield inconclusive results when the dimension equals $2k-1$. To address such cases and ensure the structural integrity of our extended codes, we follow the method in \cite{RL} and introduce an alternative, more refined criterion based on the Cauchy structure of the generator matrix.

\begin{lemma}[\cite{RL}]\label{Cauchy}
    Let $G = [I_k \mid A]$ be the standard generator matrix of an $[n, k]$ MDS code $C$. Then $C$ is equivalent to an RS code if and only if the matrix $A^{\prime} = [a_{i,j}^{-1}]$ (where $a_{i,j}^{-1}$ denotes the multiplicative inverse of $a_{i,j}$ in $\F_q$) satisfies the following conditions:
    \begin{itemize}
        \item All entries of $A$ are non-zero.
        \item All $2 \times 2$ minors of $A^{\prime}$ are non-zero.
        \item All $3 \times 3$ minors of $A^{\prime}$ are zero.
    \end{itemize}
\end{lemma}

\section{Meet-in-the-Middle Algorithm}

The name ``meet-in-the-middle'' is motivated by the similarity between our algorithm and the classical meet-in-the-middle technique in cryptography. Instead of performing a direct exhaustive search over all possible extensions, we split the verification process into two interacting parts: the geometric constraints arising from minimum weight codewords and the candidate extension vectors. The algorithm then detects valid extensions through inner product tests between these two sets, significantly reducing the search complexity.

Let $G_0 = (\boldsymbol{g}_1,\dots,\boldsymbol{g}_n)$ be a generator matrix of an $[n,k,n-k+1]$ MDS code $C_0$, where $\boldsymbol{g}_i$ denotes its $i$-th column. Our goal is to construct an extended $[n+1,k,n-k+2]$ MDS code $C_1$ by appending a single column $\boldsymbol{\beta}$ to $G_0$. Specifically, we seek a vector $\boldsymbol{\beta} \in \F_q^k$ such that the augmented matrix
\[
G_1 = (G_0 \mid \boldsymbol{\beta})
\]
generates the desired MDS code $C_1$.

For any index set $1 \le i_1 < \cdots < i_s \le n$, let $G_{i_1,\dots,i_s}$ denote the submatrix of $G_0$ consisting of these selected columns:
\[
G_{i_1,\dots,i_s} = (\boldsymbol{g}_{i_1},\dots,\boldsymbol{g}_{i_s}).
\]

\begin{lemma}\label{MainLemma}
    For an MDS $[n,k,d]$ code $C_0$ with generator matrix $G_0$, the code $C_1$ generated by $G_1$ is MDS if and only if for any $1 \le i_1 < \cdots < i_{k-1} \le n$, the matrix
    \[
    (G_{i_1,\dots,i_{k-1}} \mid \boldsymbol{\beta})
    \]
    is invertible. Moreover, if $C_0$ is not equivalent to any Reed--Solomon (RS) code, then neither is $C_1$.
\end{lemma}

\begin{proof}
    Since $C_0$ is an MDS code, any $k$ columns of $G_0$ are linearly independent. Therefore, $C_1$ is an MDS code if and only if any $k-1$ columns of $G_0$ together with the new column $\boldsymbol{\beta}$ remain linearly independent. This is equivalent to saying that the matrix $(G_{i_1,\dots,i_{k-1}} \mid \boldsymbol{\beta})$ is invertible for all  $1\leq i_1\leq \cdots\leq i_k\leq n$.

    Assume $C_1$ is equivalent to some RS code. Then $C_0$, as a punctured code of $C_1$, is also equivalent to an RS code, which yields a contradiction.
\end{proof}

Now, we formalize the procedure for finding such an extension vector $\boldsymbol{\beta}$ to construct the extended code $C_1$.

\begin{algorithm}[Meet-in-the-Middle Algorithm]
    \qquad
    \begin{enumerate}
        \item [\textbf{Step 1:}] Find all standard nonzero vectors $\boldsymbol{v} \in \F_q^k$ (here ``standard'' means the first nonzero entry is $1$) such that $\boldsymbol{v}$ is orthogonal to $k-1$ columns of $G_0$.
        \item [\textbf{Step 2:}] Put all such vectors $\boldsymbol{v}$ from Step 1 into a set $M$. Find a column vector $\boldsymbol{\beta}$ such that $\boldsymbol{\beta}$ is not orthogonal to any vector in $M$.
        \item [\textbf{Step 3:}] Construct the extended code $C_1$ generated by
        \[
        G_1 = (G_0 \mid \boldsymbol{\beta}).
        \]
    \end{enumerate}
\end{algorithm}

In Step 2, $\boldsymbol{\beta}$ is not in any hyperplane spanned by any $k-1$ columns of $G_0$ since $\boldsymbol{\beta}$ is not orthogonal to the normal vectors of these hyperplanes. Therefore, by Lemma \ref{MainLemma}, $G_1 = (G_0 \mid \boldsymbol{\beta})$ indeed generates an $[n+1,k,n-k+2]$ MDS code.

To demonstrate the computational advantage of our proposed algorithm, we estimate and compare the complexity in terms of the number of multiplications in $\F_q$.

\textbf{Algorithm 1:} We calculate the computational cost according to the steps.
\begin{itemize}
    \item \textit{Cost of Step 1:} Finding the normal vectors $\boldsymbol{v}$ requires solving $\binom{n}{k-1}$ linear systems of equations. Using Gaussian elimination on a $(k-1) \times k$ matrix takes approximately $k^3/3$ multiplications. Thus, Step 1 has about $\binom{n}{k-1} \frac{k^3}{3}$ multiplications.

    \item \textit{Cost of Step 2:} For a given candidate $\boldsymbol{\beta}\in\F_q^k$, we test the inner products $\langle \boldsymbol{v}, \boldsymbol{\beta} \rangle$ sequentially for $\boldsymbol{v} \in M$. In the worst case, this requires $k|M|$ multiplications. However, the test stops immediately once $\langle \boldsymbol{v}, \boldsymbol{\beta} \rangle = 0$ for some $\boldsymbol{v} \in M$. For a fixed $\boldsymbol{v} \in M$, the probability of $\langle \boldsymbol{v}, \boldsymbol{\beta} \rangle = 0$ is $1/q$. Hence, the expected number of tested vectors in $M$ for a single $\boldsymbol{\beta}$ is
    \[
    \sum_{i=1}^{|M|}\left(1-\frac{1}{q}\right)^{i-1} = q\left(1-\left(1-\frac{1}{q}\right)^{|M|}\right) \leq  q.
    \]
    By MacWilliams identities \cite{MS}, an $[n, k, n-k+1]$ MDS code has $A_{n-k+1}=(q-1)\binom{n}{n-k+1}$ minimum weight codewords, which implies
    \[
    |M| = \binom{n}{k-1}.
    \]
    Given $q^{k-1}$ possible choices of $\boldsymbol{\beta}$, the total expected number of multiplications for Step 2 is bounded by
    \[
    k \cdot q^{k-1} \cdot q\left(1 - \left(1 - \frac{1}{q}\right)^{\binom{n}{k-1}}\right) =O(kq^k).
    \]
    \item \textit{Cost of Step 3:} No multiplications.
\end{itemize}

Therefore, the total expected cost of our algorithm is approx:
\begin{equation}\label{kqk}
\binom{n}{k-1} \frac{k^3}{3} + O(k q^k).
\end{equation}

\textbf{Generic search:} We also consider standard vectors. For each candidate $\boldsymbol{\beta} \in \F_q^k$ and each choice of indices $1 \le i_1 < \cdots < i_{k-1} \le n$, one needs to verify that
\[
\boldsymbol{\beta} \notin \mathrm{span}_{\F_q}(\boldsymbol{g}_{i_1},\ldots,\boldsymbol{g}_{i_{k-1}}).
\]
This can be checked by computing the determinant(or rank) $\det(\boldsymbol{g}_{i_1},\ldots,\boldsymbol{g}_{i_{k-1}},\boldsymbol{\beta})$. Using Gaussian elimination, this requires approximately $k^3/3$ multiplications. Hence, the total computational cost of the naive method is at most
\begin{equation}\label{k3qk}
    \binom{n}{k-1} q^{k-1} \left(\frac{k^3}{3} + O(k^2)\right) = O\left(\binom{n}{k-1} q^{k-1} k^3\right).
\end{equation}

\textbf{Comparison:} When $n<q$, $\binom{n}{k-1}\frac{k^3}{3}$ from Step 1 is negligible. Therefore, the expected cost of our algorithm is approximately $O(k q^k)$. In contrast, the direct exhaustive search requires $O\left(\binom{n}{k-1} q^{k-1} k^3\right)$. In contrast to the exhaustive search, our approach offers a significant computational speedup.

To show our algorithm, we provide some example over finite fields. Following \cite{WLLZ}, we begin with an $[n,k]$ MDS code $C_0$ derived from an elliptic curve with the maximum number of rational points. By applying the extension procedure (Steps 1, 2, and 3), we can identify suitable extension vectors $\boldsymbol{\beta}_i$. The resulting extended matrix generates an MDS code that is strictly not equivalent to any RS code. We provide an example over $\F_8$ below.

\begin{example}[\cite{WLLZ}, Theorem 1]
    Consider $\F_8=\F_2(w)$, where $w$ satisfies $w^3+w+1=0$. Let $y^2+wxy+wy=x^3+w+1$ be an elliptic curve over $\F_8$.

    Let $D=(0:1:0)+(w^4:1:1)+(0:w^3:1)+(w^3:w^4:1)+(w^3:w:1)+(0:1:1)+(w^4:w^2:1)$ and $G=3P_0(w:w^4:1)$ be two divisors over $\F_8$. The algebraic geometry code $C=C_\mathscr{L}(D,G)$ is then a $[7,3,5]$ MDS code over $\F_8$ with generator matrix
    \[
    G_0=
    \begin{pmatrix}
        1 & 1 & 1 & 1 & 1 & 1 & 1 \\
        0 & w^4 & w^6 & w^3 & 1 & 1 & w^6 \\
        w^2 & w & 0 & w^6 & w^3 & 0 & w^3
    \end{pmatrix}.
    \]

    Transforming $G_0$ into its standard form, we get
    \[
    G_{std} =
    \begin{pmatrix}
        1 & 0 & 0 & w^6 & 1 & w^6 & w^2 \\
        0 & 1 & 0 & w^4 & w^4 & 1 & w^5 \\
        0 & 0 & 1 & w & w^4 & w^6 & w
    \end{pmatrix}.
    \]
    The corresponding matrix $A^{\prime}$ from the standard form is given by
    \[
    A^{\prime} =
    \begin{pmatrix}
        w & w^7 & w & w^5 \\
        w^3 & w^3 & w^7 & w^2 \\
        w^6 & w^3 & w & w^6
    \end{pmatrix}.
    \]
    It is easy to verify that the $3 \times 3$ minor formed by columns 1, 2, and 4 is non-zero:
    \[
    \det
    \begin{pmatrix}
        w & 1 & w^5 \\
        w^3 & w^3 & w^2 \\
        w^6 & w^3 & w^6
    \end{pmatrix}
    = w^5 \neq 0.
    \]
    So $C$ is not equivalent to any RS code by Lemma~\ref{Cauchy}.

    Applying our extension algorithm, we successively append
    \[
    \boldsymbol{\beta}_1=(1,0,w^6)^T, \quad \boldsymbol{\beta}_2=(1,w^4,w^2)^T, \quad \boldsymbol{\beta}_3=(1,w^3,w)^T.
    \]
    The final generator matrix $G_3=(G_0\mid \boldsymbol{\beta}_1,\boldsymbol{\beta}_2,\boldsymbol{\beta}_3)$ generates a $[10,3,8]$ MDS code over $\F_8$ and $dim(C^2)=6$. This attains the maximal possible length for MDS codes over $\F_8$. In particular, the resulting code is non-RS.

    The code $C_0$ comes from an elliptic curve and attains the upper bound of length proposed in \cite{HR}, our extension algorithm further enlarges the code and finally reaches the MDS bound $q+2$ over $\F_8$. Although this [10,3,8] code is not new, this example still shows that our algorithm can find MDS codes reaching the theoretical maximum length.
\end{example}

Moreover, we give an example over $\F_9$:

\begin{example}
    Consider $\F_9=\F_3(w)$, where $w$ satisfies $w^2+2w+2=0$. Let $y^2 = x^3 + 2x$ be an elliptic curve over $\F_9$.

    Let $D=(w^2 : w : 1)+(w^6 : w^3 : 1)+(0 : 0 : 1)+(w^6 : w^7 : 1)+(1 : 0 : 1)+(w^2 : w^5 : 1)+(2 : 0 : 1)$ and $G=3(0 : 1 : 0) + (w^5 : w^7 : 1)$ be two divisors over $\F_9$. The algebraic geometry code $C=C_\mathscr{L}(D,G)$ is then a $[7,4,4]$ MDS code over $\F_9$ with a generator matrix
    \[
    G_0=
    \begin{pmatrix}
        w^6 & w^5 & w^6 & w^6 & 1 & w^7 & w^5 \\
        w^7 & 0 & 2 & w^5 & w^6 & w & w^3 \\
        1 & w^5 & w^7 & w^5 & w & 1 & w^2 \\
        1 & 1 & 1 & 1 & 1 & 1 & 1
    \end{pmatrix}.
    \]

    The standard form of $G_0$ is $[I_4 \mid A]$, where the submatrix $A$ is given by
    \[
    A =
    \begin{pmatrix}
        w^7 & w^1 & w^6 \\
        w^7 & w^5 & w^0 \\
        w^3 & w^7 & w^7 \\
        w^6 & w^6 & w^4
    \end{pmatrix}.
    \]
    Taking the multiplicative inverse of each entry, we have
    \[
    A^{\prime} =
    \begin{pmatrix}
        w^1 & w^7 & w^2 \\
        w^1 & w^3 & w^0 \\
        w^5 & w^1 & w^1 \\
        w^2 & w^2 & w^4
    \end{pmatrix}.
    \]
    We have
    \[
    \det
    \begin{pmatrix}
        w^1 & w^7 & w^2 \\
        w^1 & w^3 & w^0 \\
        w^5 & w^1 & w^1
    \end{pmatrix}
    = w^6 \neq 0.
    \]
    Thus, $C$ is not equivalent to any RS code.

    Applying our extension algorithm, we successively append the following column vectors:
    \[
    \boldsymbol{\beta}_1=(0, 1, w^2, w^6)^T, \quad \boldsymbol{\beta}_2= (1, 0, w^2, w^5)^T.
    \]
    The final generator matrix $G_2=(G_0\mid \boldsymbol{\beta}_1,\boldsymbol{\beta}_2)$ yields a $[9,4,6]$ MDS code over $\F_9$ and $dim(C^2)=9$, which remains non-RS.
\end{example}

It is worth noting that any $[10,4,7]$ MDS code over $\F_9$ corresponds to a $10$-arc in $PG(3,9)$. Classical results on arcs in finite projective spaces (see \cite{H}) indicate that large complete arcs in $PG(3,9)$ are highly constrained and closely related to normal rational curves. Since RS codes are stable under puncturing, any extension of our non-RS $[7,4,4]$ code to length $10$ would require a non-RS $10$-arc in $PG(3,9)$, which is not supported by known classifications. Hence, length $n=9$ appears to be the maximal achievable value for non-RS MDS codes of dimension $4$ over $\F_9$, and our algorithm terminates exactly at this boundary.

In short, the examples over $\F_8$ and $\F_9$ show that the algorithm can go beyond the usual maximal lengths arising from elliptic-curve constructions and still reach the MDS bound. In particular, the non-RS structure is strictly preserved.

To test our method in different finite fields, we start with non-RS MDS codes over $\F_{11}$ from Example 3.7 in \cite{GZ}.  To make our discussion self-contained, we briefly introduce the TGRS codes with $l$ twists proposed by Gu and Zhang. Let $n$, $k$, and $l$ be the code length, dimension, and the number of twists, respectively. A twisted polynomial space $\mathcal{S}$ is generated by perturbing $l$ basis elements of the standard space. Specifically, for a twist parameter vector $\boldsymbol{\eta} = (\eta_{1}, \dots, \eta_{l}) \in \F_{q}^{l}$ and some integers $h_j \ge k$, the basis of $\mathcal{S}$ is obtained by replacing $x^{t_j}$ with $x^{t_j} + \eta_j x^{h_j}$ for $j=1, \dots, l$.

By evaluating the functions of $\mathcal{S}$ at distinct points $\boldsymbol{\alpha} = (\alpha_{1}, \dots, \alpha_{n}) \in \F_{q}^{n}$, we obtain the twisted GRS (TGRS) code $\mathcal{C} = \mathrm{ev}_{\boldsymbol{\alpha},1}(\mathcal{S})$. For any $k$-subset $I \subseteq \{1, 2, \dots, n\}$, let $M(\boldsymbol{\eta}, \boldsymbol{\alpha}, I, l)$ denote a specific determinant associated with the evaluation points and twist parameters over $I$. The fundamental result by Gu and Zhang states that $\mathcal{C}$ is MDS if and only if $M(\boldsymbol{\eta}, \boldsymbol{\alpha}, I, l) \neq 0$ for all possible subsets $I$, which is reviewed in the following result.

\begin{lemma} [\cite{GZ}, Theorem 3.3]
    Suppose that $3 \le k < n$ and $\alpha_{1}, \alpha_{2}, \dots, \alpha_{n}$ are distinct elements of $\F_{q}$. Let $\Omega$ be the set of all $\boldsymbol{\eta} \in \F_{q}^{l}$ such that for each $k$-subset $I$, we have $M(\boldsymbol{\eta}, \boldsymbol{\alpha}, I, l) \neq 0$. Then $\mathcal{C} = \mathrm{ev}_{\boldsymbol{\alpha},1}(\mathcal{S})$ is MDS if and only if $\boldsymbol{\eta} = (\eta_{1}, \dots, \eta_{l}) \in \Omega$.
\end{lemma}

\begin{example} [\cite{GZ}, Example 3.7]
    Let $\boldsymbol{\alpha}=(1, 2, 3, 5, 6, 8, 9, 10)$ and $\boldsymbol{\eta} = (\eta_1, \eta_2) = (2, 9)$. Then $\mathcal{S} = \{1, x + 2x^3, x^2 + 9x^4\}$ and $C$ is an $[8,3,6]$ MDS code over $\F_{11}$ generated by
    \[
    G_0 = \begin{pmatrix}
        1 & 1 & 1 & 1 & 1 & 1 & 1 & 1 \\
        3 & 7 & 2 & 2 & 9 & 9 & 4 & 8 \\
        10 & 5 & 1 & 7 & 7 & 1 & 5 & 10
    \end{pmatrix}.
    \]

    We first show that $C$ is not equivalent to any RS code. A direct computation shows that the Schur square $C^2$ has  $\dim(C^2)=6$. Hence, $C$ cannot be equivalent to any RS code.

    By applying our extension algorithm, we successively append the vector
    \[
    \boldsymbol{\beta}_1=(1,0,2)^T.
    \]
    $G_1=(G_0\mid \boldsymbol{\beta}_1)$ generates a $[9,3,7]$ non-RS MDS code and $dim(C^2)=6$.
\end{example}

\begin{example} [\cite{GZ}, Example 3.7]
    Let $\boldsymbol{\alpha} = (1, 2, 3, 5, 6, 8, 9, 10)$ and $\boldsymbol{\eta} = (\eta_1, \eta_2) = (4, 4)$. Then the twisted polynomial space is $\mathcal{S} = \{1, x, x^2 + 4x^4, x^3 + 4x^5\}$. This yields an $[8,4,5]$ MDS code $C$ over $\F_{11}$ generated by
    \[
    G_0 = \begin{pmatrix}
        1 & 1 & 1 & 1 & 1 & 1 & 1 & 1 \\
        1 & 2 & 3 & 5 & 6 & 8 & 9 & 10 \\
        5 & 2 & 3 & 6 & 6 & 3 & 2 & 5 \\
        5 & 4 & 9 & 8 & 3 & 2 & 7 & 6
    \end{pmatrix}.
    \]

    For this code, the Schur product criterion is inconclusive since $\dim(C^2) = 7 = 2k-1$. Therefore, we apply Lemma \ref{Cauchy}. The standard form of $G_0$ is $[I_4 \mid A]$, where $A$ is given by
    \[
    A = \begin{pmatrix}
        2 & 2 & 2 & 1 \\
        10 & 9 & 4 & 6 \\
        8 & 3 & 10 & 3 \\
        3 & 9 & 7 & 2
    \end{pmatrix}.
    \]
    Consequently, the corresponding inverse matrix $A^{\prime}$ is
    \[
    A^{\prime} = \begin{pmatrix}
        6 & 6 & 6 & 1 \\
        10 & 5 & 3 & 2 \\
        7 & 4 & 10 & 4 \\
        4 & 5 & 8 & 6
    \end{pmatrix}.
    \]
    We have $\operatorname{rank}(A^{\prime}) = 4$. Since the rank is strictly greater than 2, it guarantees that not all $3 \times 3$ minors of $A^{\prime}$ are zero. Thus, by Lemma~\ref{Cauchy}, $C$ is not equivalent to any RS code.

    By applying our extension algorithm, we successively append the vector
    \[
    \boldsymbol{\beta}_1=(1,0,8,0)^T.
    \]
    $G_1=(G_0\mid \boldsymbol{\beta}_1)$ generates a $[9,4,6]$ non-RS MDS code and $dim(C^2)=8$.
\end{example}

In these two examples, we apply our method to the TGRS codes introduced in \cite{GZ} and extend them to longer non-RS MDS codes.

We can also start with a column-twisted MDS code over $\F_{13}$ based on \cite{LLWZ}. To make our discussion self-contained, we briefly introduce column twisted Reed-Solomon (CTRS) codes. Choose $n-1$ distinct evaluation points $a_1, a_2, \dots, a_{n-1} \in \F_{q}$, and let $b, c \in \F_{q}$ with a non-zero parameter $\lambda \in \F_{q}^*$.

For any element $x \in \F_{q}$, let $\boldsymbol{v}(x) = (1, x, x^2, \dots, x^{k-1})^T$. The CTRS code $C_1(b, c)$ is defined by the following $k \times n$ generator matrix
\[
G =
\begin{pmatrix}
    1 & 1 & \dots & 1 & 1 - \lambda \\
    a_1 & a_2 & \dots & a_{n-1} & b - \lambda c \\
    \vdots & \vdots & \ddots & \vdots & \vdots \\
    a_1^{k-1} & a_2^{k-1} & \dots & a_{n-1}^{k-1} & b^{k-1} - \lambda c^{k-1}
\end{pmatrix}.
\]

For CTRS $C_1(b, c)$, we have the following result

\begin{lemma} [\cite{LLWZ}, Lemma 1]
    The code $C_1(b, c)$ generated by $G$ is MDS if and only if for any $1 \le i_1 < i_2 < \dots < i_{k-1} \le n-1$,
    \[
    \prod_{j=1}^{k-1} (b - a_{i_j}) - \lambda \prod_{j=1}^{k-1} (c - a_{i_j}) \neq 0.
    \]
\end{lemma}

\begin{example}[\cite{LLWZ}, Theorem~2]
Let $\gamma=2$ in $\F_{19}$. Let $H=\langle\gamma^2\rangle$, $k=5$, $b=0$, $c=1$, and $(\lambda_1,\lambda_2)=(2,8)$. Taking the elements $\mu_i$ of $H\setminus\{1\}$ in increasing order and setting $a_i=(b-\mu_i c)/(1-\mu_i)$ gives
\[
\boldsymbol\alpha=(14,6,5,17,13,3,15,7).
\]
With $\boldsymbol v(x)=(1,x,x^2,x^3,x^4)^T$, the first eight columns of $G_0$ are $\boldsymbol v(a_i)$, and its last two columns are $\boldsymbol v(b)-\lambda_j\boldsymbol v(c)$, $j=1,2$. Explicitly,
\[
G_0=\begin{bmatrix}
1&1&1&1&1&1&1&1&18&12\\
14&6&5&17&13&3&15&7&17&11\\
6&17&6&4&17&9&16&11&17&11\\
8&7&11&11&12&8&12&1&17&11\\
17&4&17&16&4&5&9&7&17&11
\end{bmatrix}.
\]
$C_0$ is $[10,5,6]$ MDS and $\dim(C_0^2)=10$, so $C_0$ is non-RS.

Our extension search finds the unique admissible projective column
\[
\boldsymbol\beta_1=(0,0,0,0,1)^T.
\]
Thus $G_1=(G_0\mid\boldsymbol\beta_1)$ generates a $[11,5,7]$ MDS code $C_1$ over $\F_{19}$ with $\dim(C_1^2)=10$.
\end{example}

\begin{example} [\cite{LLWZ},  Lemma~1(1)]
    Let $w=2$ in $\F_{13}$. Let $\{a_1, a_2, \dots, a_6\}=\{0,1,2,3,4,5\}$, $b=6$, $c=12$, and $\lambda=1$. Then $C$ is a $[7,4,4]$ MDS code over $\F_{13}$ generated by
    \[
    G_0=
    \begin{pmatrix}
        1 & 1 & 1 & 1 & 1 & 1 & 0\\
        0 & 1 & 2 & 3 & 4 & 5 & 7\\
        0 & 1 & 4 & 9 & 3 & 12 & 9\\
        0 & 1 & 8 & 1 & 12 & 8 & 9
    \end{pmatrix}.
    \]
    We first show that $C$ is not equivalent to any RS code. A direct computation gives $\dim(C^2)=7=2k-1$, so the Schur-square criterion is inconclusive. We therefore apply Lemma~\ref{Cauchy}. Reducing $G_0$ to systematic form gives $[I_4\mid A]$, where
    \[
    A=
    \begin{bmatrix}
        12&9&12\\
        4&2&3\\
        7&6&3\\
        4&10&8
    \end{bmatrix}.
    \]
    Then
    \[
    A'=
    \begin{bmatrix}
        12&3&12\\
        10&7&9\\
        2&11&9\\
        10&4&5
    \end{bmatrix}.
    \]
    The determinant of the submatrix formed by the first three rows of $A'$ is $10\ne0$ in $\F_{13}$. Hence $\operatorname{rank}(A')=3>2$, and Lemma~\ref{Cauchy} shows that $C$ is not equivalent to any RS code. Hence, $C$ cannot be equivalent to any RS code.

    By applying our extension algorithm, we successively append projective columns. We append the vectors
    \[
    \boldsymbol{\beta}_1=(1,11,2,0)^T, \boldsymbol{\beta}_2=(1,11,7,1)^T.
    \]
    Then $G_2=(G_0\mid\boldsymbol{\beta}_1,\boldsymbol{\beta}_2)$ generates a $[9,4,6]$ non-RS MDS code over $\F_{13}$ with $\dim(C^2)=9$. Moreover, $\dim(C_2^2)>2k$. Since the Schur square of a one-column CTRS code $C_1(b,c)$ as defined in \cite{LLWZ} has dimension $2k$, $C_2$ is not equivalent to any CTRS code. Its length $9$ also exceeds $(13+3)/2=8$, the largest length attained by the subgroup-based construction in \cite{LLWZ} over $\F_{13}$.
\end{example}

Next, we apply our extension algorithm to twisted Reed--Solomon codes. We first recall the one-twist construction of \cite{BPR}. Let $1\le k<n\le q$, choose distinct evaluation points $\boldsymbol{\alpha}=(\alpha_1,\ldots,\alpha_n)\in\F_q^n$, and take
\[
1\le t\le n-k, 0\le h\le k-1, \eta\in\F_q^*.
\]
Here $t$ is the twist, $h$ is the hook, and $\eta$ is the twist coefficient. Define
\[
\mathcal{P}_{t,h,\eta}^{n,k}=\left\{\sum_{i=0}^{k-1}f_i x^i+\eta f_hx^{k-1+t}:f_0,\ldots,f_{k-1}\in\F_q\right\}.
\]
The corresponding twisted Reed--Solomon code is
\[
\mathcal{C}_{\boldsymbol{\alpha},t,h,\eta}^{n,k}=\left\{(f(\alpha_1),\ldots,f(\alpha_n)):f\in\mathcal{P}_{t,h,\eta}^{n,k}\right\}.
\]

We use the case $t=1$ and $h=0$, for which
\[
\mathcal{P}_{1,0,\eta}^{n,k}=\operatorname{span}_{\F_q}\{1+\eta x^k,x,x^2,\ldots,x^{k-1}\}.
\]
Following Definition~3, a code of this form is called a \emph{$(*)$-twisted Reed--Solomon code} if there is a proper subgroup $H$ of the multiplicative group $\F_q^*$ such that
\[
\alpha_i\in H\cup\{0\}\quad(1\le i\le n), (-1)^k\eta^{-1}\notin H.
\]

\begin{lemma}[\cite{BPR}, Corollary~27]\label{lem:BPR-star}
Every $(*)$-twisted code $\mathcal{C}_{\boldsymbol{\alpha},1,0,\eta}^{n,k}$ is an $[n,k,n-k+1]$ MDS code. If $3\le k<n/2$, it is not equivalent to any RS code.
\end{lemma}

For all the $(*)$-twisted examples below, $q$ is odd. Let $\gamma$ be a primitive element of $\F_q$, and choose
\[
H=\langle\gamma^2\rangle=\{\gamma^{2j}:0\le j<(q-1)/2\}.
\]
This is the subgroup of all nonzero squares in $\F_q$ and has order $(q-1)/2$. Taking
\[
T=H\cup\{0\}, \eta=(-1)^k\gamma^{-1}
\]
gives $n=|T|=(q+1)/2$ and $(-1)^k\eta^{-1}=\gamma\notin H$, so Lemma~\ref{lem:BPR-star} applies. The vector $\boldsymbol{\alpha}$ specifies the order in which the elements of $T$ are evaluated; this order is also the column order of the generator matrix.

Moreover, when $3\le k\le(q-1)/2-2$, the length $(q+1)/2$ is the largest possible among one-twist MDS codes with $t=1$, $h=0$, and evaluation points in $\F_q$.

\begin{example}[\cite{BPR}, Theorems~13]
Let $\gamma=3$ in $\F_{17}$, and set $H=\langle\gamma^2\rangle, \eta=\gamma^{-1}=6, \boldsymbol{\alpha}=(0,1,9,13,15,16,8,4,2)$. For $k=6$, evaluating $\mathcal P=\operatorname{span}_{\F_{17}}\{1+6x^6,x,x^2,x^3,x^4,x^5\}$ at $\boldsymbol{\alpha}$ gives a $(*)$-twisted $[9,6,4]$ MDS code $C_0$ with generator matrix
\[
G_0=
\begin{bmatrix}
1&7&8&12&11&7&8&12&11\\
0&1&9&13&15&16&8&4&2\\
0&1&13&16&4&1&13&16&4\\
0&1&15&4&9&16&2&13&8\\
0&1&16&1&16&1&16&1&16\\
0&1&8&13&2&16&9&4&15
\end{bmatrix}.
\]
Here $\dim(C_0^2)=9$. This does not means $C_0$ is not equivalent to some RS code. Instead, consider $D_0=C_0^\perp$. It is an $[9,3,7]$ MDS code,
with the following generator matrix:
\[
H_0=
\begin{bmatrix}
1&15&5&2&1&7&1&0&0\\
5&14&14&8&12&3&0&1&0\\
3&6&16&8&5&8&0&0&1
\end{bmatrix},
\qquad G_0H_0^T=0.
\]
We have $\dim(D_0^2)=6>5=2\dim(D_0)-1$. Thus $D_0$ is not equivalent to a RS code. Since the dual of a RS code is again a RS code up to equivalence, $C_0$ is also non-RS.

Applying our extension algorithm to $C_0$, we append
\[
\boldsymbol{\beta}_1=(1,3,16,1,16,13)^T.
\]
Then $G_1=(G_0\mid\boldsymbol{\beta}_1)$ generates a $[10,6,5]$ MDS code $C_1$. And we have $\dim(C_1^2)=10$. Similarly, let $D_1=C_1^\perp$. Then $D_1$ is a $[10,4,7]$ MDS code. Set $\boldsymbol h=(15,2,9,5,4,7,0,0,0)$. A generator matrix of $D_1$ is
\[
H_1=
\begin{bmatrix}
    H_0 & \mathbf{0}_{3\times1}\\
    \boldsymbol h & 1
\end{bmatrix}.
\]
We have $\dim(D_1^2)=10>7=2\dim(D_1)-1$. Consequently, $D_1$, and hence $C_1$, is non-RS. Moreover, $\dim(D_1^2)>2\dim(D_1)$, whereas the CTRS codes
constructed in \cite{LLWZ} have Schur-square dimension $2k$. Thus, $D_1$ is not equivalent to any code in those CTRS families. Its length matches the largest length attained by the construction of \cite{LLWZ} over $\F_{17}$.
\end{example}

\begin{example}[\cite{BPR}, Corollary~27]
Let $\gamma=3$, $\eta=6$ and $T=(0,1,9,13,15,16,8,4,2)$ in $\F_{17}$. Evaluating $\operatorname{span}_{\mathbb{F}_{17}}\{1+\eta x^4,x,x^2,x^3\}$ on $T$ gives the following code. Its length is the largest possible within the affine one-twist family with $t=1,h=0$, by Theorem~13 of \cite{BPR}. $C_0$ is a $[9,4,6]$ code over $\mathbb{F}_{17}$ with generate matrix
\[
G_0=\begin{bmatrix}
1 & 7 & 12 & 7 & 12 & 7 & 12 & 7 & 12 \\
0 & 1 & 9 & 13 & 15 & 16 & 8 & 4 & 2 \\
0 & 1 & 13 & 16 & 4 & 1 & 13 & 16 & 4 \\
0 & 1 & 15 & 4 & 9 & 16 & 2 & 13 & 8
\end{bmatrix}.
\]
The Schur square satisfies $\dim(C_0^2)=8>7$, so $C_0$ is non-RS. Applying our extension algorithm, we successively append
\begin{align*}
\boldsymbol{\beta}_{1}&=(0,1,0,6)^T
\end{align*}
The matrix $G_{1}=(G_0\mid\boldsymbol{\beta}_1)$ generates a $[10,4,7]$ MDS code over $\mathbb{F}_{17}$, with $\dim(C_{1}^2)=9$. Similarly as pervious example, this code is not equivalent to any code in those CTRS families. Its length also matches the largest length attained by the construction of \cite{LLWZ} over $\F_{17}$.
\end{example}

We next take $k=3$ in the same multiplicative-subgroup construction. The common choice of the twist coefficient now becomes $\eta=-\gamma^{-1}$, and the initial code is
\[
C_0=\mathcal{C}_{\boldsymbol{\alpha},1,0,\eta}^{n,3}, \mathcal{P}_{1,0,\eta}^{n,3}=\operatorname{span}_{\F_q}\{1+\eta x^3,x,x^2\}.
\]
In each of the following examples, the evaluation set remains $T=H\cup\{0\}$, with its elements ordered as specified by $\boldsymbol{\alpha}$. For these dimension-three codes, Schur-square dimension is $6$ which means they are not equivalent to any RS codes.

\begin{example}[\cite{BPR}, Corollary~27]
Let $\gamma=2$ in $\F_{13}$, and set $\eta=6, \boldsymbol{\alpha}=(0,1,3,4,9,10,12)$. The resulting $[7,3,5]$ MDS code $C_0$ over $\F_{13}$ is generated by
\[
G_0=
\begin{bmatrix}
            1&7&7&8&7&8&8\\
            0&1&3&4&9&10&12\\
            0&1&9&3&3&9&1
        \end{bmatrix}.
\]
A direct computation gives $\dim(C_0^2)=6$, so $C_0$ is non-RS.

Applying our extension algorithm, we successively append the following column vectors:
\[
\boldsymbol{\beta}_1=(1,1,0)^T, \boldsymbol{\beta}_2=(1,5,3)^T, \boldsymbol{\beta}_3=(1,7,3)^T.
\]
Thus $G_3=(G_0\mid\boldsymbol{\beta}_1,\boldsymbol{\beta}_2, \boldsymbol{\beta}_3)$ generates a $[10,3,8]$ MDS code $C_3$ over $\F_{13}$. Moreover, $\dim(C_3^2)=6$, so the extended code is again non-RS.
\end{example}

\begin{example}[\cite{BPR}, Corollary~27]
Let $\gamma=2$, and set $\eta=9, \boldsymbol{\alpha}=(0,1,4,16,7,9,17,11,6,5)$ in $\F_{19}$. The resulting $[10,3,8]$ MDS code $C_0$ over $\F_{19}$ is generated by
\[
G_0=
\begin{bmatrix}
            1&10&7&5&10&7&5&10&7&5\\
            0&1&4&16&7&9&17&11&6&5\\
            0&1&16&9&11&5&4&7&17&6
        \end{bmatrix}.
\]
A direct computation gives $\dim(C_0^2)=6$, so $C_0$ is non-RS.

Applying our extension algorithm, we successively append the following column vectors:
\[
\boldsymbol{\beta}_1=(1,5,3)^T, \boldsymbol{\beta}_2=(1,16,14)^T, \boldsymbol{\beta}_3=(1,17,2)^T, \boldsymbol{\beta}_4=(0,0,1)^T.
\]
Thus $G_4=(G_0\mid\boldsymbol{\beta}_1,\boldsymbol{\beta}_2, \boldsymbol{\beta}_3,\boldsymbol{\beta}_4)$ generates a $[14,3,12]$ MDS code $C_4$ over $\F_{19}$. Moreover, $\dim(C_4^2)=6$, so the extended code is again non-RS.
\end{example}
\begin{example}[\cite{BPR}, Corollary~27]
Let $\gamma=5$ in $\F_{23}$. Set $\eta=9, \boldsymbol{\alpha}=(0,1,2,4,8,16,9,18,13,3,6,12)$. The resulting $[12,3,10]$ MDS code $C_0$ over $\F_{23}$ is generated by

\[
G_0=
\left[
\begin{array}{*{12}{c}}
            1&10&4&2&9&19&7&3&17&14&13&5\\
            0&1&2&4&8&16&9&18&13&3&6&12\\
            0&1&4&16&18&3&12&2&8&9&13&6
\end{array}
\right].
\]
A direct computation gives $\dim(C_0^2)=6$, so $C_0$ is non-RS.

Applying our extension algorithm, we successively append the following column vectors:
\[
\boldsymbol{\beta}_1=(1,11,2)^T, \boldsymbol{\beta}_2=(1,17,3)^T.
\]
Thus $G_2=(G_0\mid\boldsymbol{\beta}_1,\boldsymbol{\beta}_2)$ generates a $[14,3,12]$ MDS code $C_2$ over $\F_{23}$. Moreover, $\dim(C_2^2)=6$, so the extended code is again non-RS.
\end{example}

\begin{example}[\cite{BPR}, Corollary~27]
Let $\F_{27}=\F_3(a)$, where $a^3+2a+1=0$ and $a$ is primitive. Taking $\gamma=a$, set $\eta=a^{12}, \boldsymbol{\alpha}=(0,1,a^2,a^4,\ldots,a^{24})$. The resulting $[14,3,12]$ MDS code $C_0$ over $\F_{27}$ is generated by

\[
G_0=
\left[
\begin{array}{*{14}{c}}
1&a^2&a^7&a^{19}&a^{18}&a^6&a^{22}&a^{14}&a^{21}&a^{15}&a^{16}&a^5&a^{13}&a^{11}\\
0&1&a^2&a^4&a^6&a^8&a^{10}&a^{12}&a^{14}&a^{16}&a^{18}&a^{20}&a^{22}&a^{24}\\
0&1&a^4&a^8&a^{12}&a^{16}&a^{20}&a^{24}&a^2&a^6&a^{10}&a^{14}&a^{18}&a^{22}
\end{array}
\right].
\]

A direct computation gives $\dim(C_0^2)=6$, so $C_0$ is non-RS.

Applying our extension algorithm, we successively append the following column vectors:
\[
\boldsymbol{\beta}_1=(1,0,a^{18})^T, \boldsymbol{\beta}_2=(1,a^9,a^{18})^T.
\]
Thus $G_2=(G_0\mid\boldsymbol{\beta}_1,\boldsymbol{\beta}_2)$ generates a $[16,3,14]$ MDS code $C_2$ over $\F_{27}$. Moreover, $\dim(C_2^2)=6$, so the extended code is again non-RS.
\end{example}

\begin{example}[\cite{BPR}, Corollary~27]
Let $\gamma=2$ in $\F_{29}$. Set $\eta=14, \boldsymbol{\alpha}=(0,1,4,16,6,24,9,7,28,25,13,23,5,20,22)$. The resulting $[15,3,13]$ MDS code $C_0$ over $\F_{29}$ is generated by
\[
G_0=
\left[
\begin{array}{*{15}{c}}
1&15&27&12&9&20&28&18&16&4&19&22&11&3&13\\
0&1&4&16&6&24&9&7&28&25&13&23&5&20&22\\
0&1&16&24&7&25&23&20&1&16&24&7&25&23&20
\end{array}
\right].
\]
A direct computation gives $\dim(C_0^2)=6$, so $C_0$ is non-RS.

Applying our extension algorithm, we successively append
\[
\boldsymbol{\beta}_1=(1,8,24)^T.
\]
Thus $G_1=(G_0\mid\boldsymbol{\beta}_1)$ generates a $[16,3,14]$ MDS code $C_1$ over $\F_{29}$. Moreover, $\dim(C_1^2)=6$, so the extended code is again non-RS.
\end{example}

We now turn to $(+)$-twisted codes with $k=3$, $t=1$, and $h=2$. Their polynomial space is
\[
\mathcal{P}_{1,2,\eta}^{n,3}=\operatorname{span}_{\F_q}\{1,x,x^2+\eta x^3\}.
\]
Choose the evaluation points from an additive subgroup $V\subsetneq\F_q$ and require $\eta^{-1}\notin V$. The resulting code is MDS by Theorem~15 and is non-RS when $n>6$ by Corollary~28. Moreover, Remark~16 permits adjoining $f(\infty)$, defined as the coefficient of $x^3$ in $f$, while preserving the MDS property. In the displayed basis, this appends the column $(0,0,\eta)^T$.

\begin{example}[\cite{BPR}, Corollary~28]
Let $\F_{16}=\F_2(a)$, where $a^4+a+1=0$, and set
\[
V=\operatorname{span}_{\F_2}\{1,a,a^2\}, \eta=a^{12}.
\]
Since $\eta^{-1}=a^3\notin V$, evaluating $\mathcal{P}_{1,2,\eta}^{8,3}$ at
\[
\boldsymbol{\alpha}=(0,1,a,a^2,a^4,a^8,a^5,a^{10})
\]
gives an $[8,3,6]$ MDS code $C_0$ over $\F_{16}$ generated by
\[
G_0=
\begin{bmatrix}
            1&1&1&1&1&1&1&1\\
            0&1&a&a^2&a^4&a^8&a^5&a^{10}\\
            0&a^{11}&a^8&a^7&a^{12}&a^{11}&a^3&a^{14}
        \end{bmatrix}.
\]
A direct computation gives $\dim(C_0^2)=6$, so $C_0$ is non-RS.

Applying our extension algorithm, we successively append the following column vectors:
\[
\boldsymbol{\beta}_1=(1,0,a^7)^T, \boldsymbol{\beta}_2=(1,a^2,a^3)^T.
\]
Thus $G_2=(G_0\mid\boldsymbol{\beta}_1,\boldsymbol{\beta}_2)$ generates a $[10,3,8]$ MDS code $C_2$ over $\F_{16}$. Moreover, $\dim(C_2^2)=6$, so the extended code is again non-RS.
\end{example}

\begin{example}[\cite{BPR}, Corollary~28]
Let $\F_{32}=\F_2(a)$, where $a^5+a^2+1=0$ and $a$ is primitive. Set $V=\operatorname{span}_{\F_2}\{1,a,a^2,a^3\}, \eta=a^{27}$. Since $\eta^{-1}=a^4\notin V$, evaluating $\mathcal{P}_{1,2,\eta}^{16,3}$ at all elements of $V$ and adjoining the evaluation at infinity gives a $[17,3,15]$ MDS code $C_0$ over $\F_{32}$ generated by
\[
\setlength{\arraycolsep}{4pt}
G_0=
\left[
\begin{array}{*{17}{c}}
1&1&1&1&1&1&1&1&1&1&1&1&1&1&1&1&0\\
0&a^3&a^2&a^{20}&a&a^6&a^{19}&a^{12}&1&a^{29}&a^5&a^8&a^{18}&a^{27}&a^{11}&a^{23}&0\\
0&a^{23}&a^7&a^{18}&a^{28}&a^{17}&1&a^{13}&a^6&a^{17}&a^{28}&a^{26}&a^{18}&a^4&a^{13}&a^{26}&a^{27}
\end{array}
\right].
\]
A direct computation gives $\dim(C_0^2)=6$, so $C_0$ is non-RS.

Applying our extension algorithm, we successively append the following column vectors:
\[
\boldsymbol{\beta}_1=(1,a^4,0)^T, \boldsymbol{\beta}_2=(0,1,a^4)^T.
\]
Thus $G_2=(G_0\mid\boldsymbol{\beta}_1,\boldsymbol{\beta}_2)$ generates a $[19,3,17]$ MDS code $C_2$ over $\F_{32}$. Moreover, $\dim(C_2^2)=6$, so the extended code is again non-RS.
\end{example}

In \cite{JMXZ}, non-RS MDS codes are constructed by selecting polynomial spaces with omitted monomials. For $I=\{0,1,\ldots,k-2,k\}$, define
\[
V_I=\operatorname{span}_{\F_q}\{x^i:i\in I\}, C(T,I)=\operatorname{ev}_{\boldsymbol{\alpha}}(V_I),
\]
where $T=\{\alpha_1,\ldots,\alpha_n\}\subseteq\F_q$. Thus $x^{k-1}$ is replaced by $x^k$, giving the generator matrix
\[
G_I=
\begin{bmatrix}
1 & \cdots & 1\\
\vdots & \ddots & \vdots\\
\alpha_1^{k-2} & \cdots & \alpha_n^{k-2}\\
\alpha_1^k & \cdots & \alpha_n^k
\end{bmatrix}.
\]

\begin{lemma}[\cite{JMXZ}, Proposition~V.1]
Suppose $3\le k\le n/2$. If
\[
\sum_{j=1}^{k}\alpha_{i_j}\ne0, \text{for every }1\le i_1<\cdots<i_k\le n,
\]
then $C(T,I)$ is an $[n,k,n-k+1]$ non-RS MDS code. In particular, over an odd prime field $\F_p$, the choice $T=\{0,1,\ldots,n-1\}$ satisfies this condition whenever
\[
2k\le n<\frac{p}{k}+\frac{k+1}{2}.
\]
\end{lemma}

Taking $p=13$, $k=3$, and $T=\{0,1,2,3,4,5\}$ gives the $[6,3,4]_{13}$ code below.

\begin{example}[\cite{JMXZ}, Corollary~V.2]
    Let $T=\{0,1,2,3,4,5\}\subseteq\F_{13}$, and take the evaluation-polynomial space $\operatorname{span}_{\F_{13}}\{1,x,x^3\}$. Then $C_0$ is a $[6,3,4]$ MDS code over $\F_{13}$ generated by
    \[
        G_0=
        \begin{pmatrix}
            1&1&1&1&1&1\\
            0&1&2&3&4&5\\
            0&1&8&1&12&8
        \end{pmatrix}.
    \]
    A direct computation gives $\dim(C_0^2)=6$. Hence, $C_0$ is not equivalent to any RS code.

    Applying our extension algorithm successively gives
    \[
        \boldsymbol{\beta}_1=(1,0,3)^T, \boldsymbol{\beta}_2=(1,1,0)^T, \boldsymbol{\beta}_3=(1,3,5)^T,\boldsymbol{\beta}_4=(1,4,5)^T.
    \]
    Therefore, $G_4=(G_0\mid \boldsymbol{\beta}_1,\boldsymbol{\beta}_2, \boldsymbol{\beta}_3,\boldsymbol{\beta}_4)$ generates a $[10,3,8]$ non-RS MDS code $C_4$ over $\F_{13}$, and $\dim(C_4^2)=6$.
\end{example}

Besides, cyclic MDS codes are constructed in \cite{XT}. Let $\zeta_n$ be a primitive $n$th root of unity, put $K=\mathbb{Q}(\zeta_n)$, and choose $J=\{j_1<\cdots<j_k\}\subseteq\{0,\ldots,n-1\}$. Set
\[
\widetilde G_J=\bigl[\zeta_n^{\,j_i(s-1)}\bigr]_{  1\le i\le k,\ 1\le s\le n}.
\]
For each $k$-subset $S$ of its columns, let $\Delta_S$ be the corresponding determinant. Assuming every $\Delta_S\ne0$, define
\[
\mathcal{B}_{n,J}=\{p:p\ prime, p\mid n|N_{K/\mathbb{Q}}(\Delta_S)|\text{ for some }S\},
\]
where $N_{K/\mathbb{Q}}$ is the algebraic norm. For a prime ideal $\mathfrak p$ of $\mathbb{Z}[\zeta_n]$ above $p\notin\mathcal{B}_{n,J}$, identify its residue field with $\F_q$ and write $\alpha=\zeta_n\bmod\mathfrak p$. The reduced generator matrix is
\[
G_J=\bigl[\alpha^{\,j_i(s-1)}\bigr]_{  1\le i\le k,\ 1\le s\le n}.
\]

\begin{lemma}[\cite{XT}, Theorem~11]
Under the preceding assumptions, for $n\ge4$ and $k\ge3$, the matrix $G_J$ generates an $[n,k,n-k+1]$ cyclic MDS code over $\F_q$. If additionally $n\ge7$, $j_k\le(n-1)/2$, and $J$ is not an arithmetic progression, this code is non-RS.
\end{lemma}

For $J=\{0,1,3\}$, Example~12 and Table~IV give $\mathcal{B}_{7,J}=\{2,7\}$ and $\mathcal{B}_{8,J}=\{2,3\}$. Taking $(n,q)=(7,29)$ and $(8,17)$ therefore yields the $[7,3,5]_{29}$ and $[8,3,6]_{17}$ codes below.

\begin{example}[\cite{XT}, Example~12]
    Let $\alpha=2$ in $\F_{17}$, and take the defining set $J=\{0,1,3\}$. Then $C_0$ generated by $G_J$ is an $[8,3,6]$ cyclic MDS code over $\F_{17}$ generated by
    \[
        G_0=
        \begin{pmatrix}
            1&1&1&1&1&1&1&1\\
            1&2&4&8&16&15&13&9\\
            1&8&13&2&16&9&4&15
        \end{pmatrix}.
    \]
    A direct computation gives $\dim(C_0^2)=6$. Hence, $C_0$ is not equivalent to any RS code.

    Applying our extension algorithm successively gives
    \[
    \begin{aligned}
        \boldsymbol{\beta}_1&=(1,1,2)^T, \boldsymbol{\beta}_2&=(1,2,1)^T, \boldsymbol{\beta}_3&=(1,9,13)^T, \boldsymbol{\beta}_4&=(1,13,9)^T.
    \end{aligned}
    \]
    Therefore, $G_4=(G_0\mid \boldsymbol{\beta}_1,\boldsymbol{\beta}_2,\boldsymbol{\beta}_3,\boldsymbol{\beta}_4)$ generates a $[12,3,10]$ non-RS MDS code $C_4$ over $\F_{17}$, and $\dim(C_4^2)=6$.
\end{example}

\begin{example}[\cite{XT}, Example~12]
    Let $\alpha=16$ in $\F_{29}$, so that $\alpha$ has multiplicative order $7$, and take the defining set $J=\{0,1,3\}$. Then $C_0$ generated by $G_J$ is a $[7,3,5]$ cyclic MDS code over $\F_{29}$ generated by
    \[
        G_0=
        \begin{pmatrix}
            1&1&1&1&1&1&1\\
            1&16&24&7&25&23&20\\
            1&7&20&24&23&16&25
        \end{pmatrix}.
    \]
    A direct computation gives $\dim(C_0^2)=6$. Hence, $C_0$ is not equivalent to any RS code.

    Applying our extension algorithm successively gives
    \[
    \begin{aligned}
        \boldsymbol{\beta}_1&=(1,0,0)^T, &
        \boldsymbol{\beta}_2&=(1,0,1)^T, &
        \boldsymbol{\beta}_3&=(1,1,0)^T,\\
        \boldsymbol{\beta}_4&=(1,2,10)^T, &
        \boldsymbol{\beta}_5&=(1,2,23)^T, &
        \boldsymbol{\beta}_6&=(1,4,24)^T,\\
        \boldsymbol{\beta}_7&=(1,7,3)^T, &
        \boldsymbol{\beta}_8&=(1,13,11)^T, &
        \boldsymbol{\beta}_9&=(1,28,8)^T.
    \end{aligned}
    \]
    Therefore, $G_9=(G_0\mid\boldsymbol{\beta}_1,\boldsymbol{\beta}_2,\ldots,\boldsymbol{\beta}_9)$  generates a $[16,3,14]$ non-RS MDS code $C_9$ over $\F_{29}$, and $\dim(C_9^2)=6$.
\end{example}

 All the source codes in SAGE of the above examples can be found in \cite{A}. The examples above illustrate the applicability of the hyperplane-based extension procedure to several families of non-RS MDS codes, including elliptic-curve codes, twisted Reed--Solomon codes, column-twisted codes, missing-monomial evaluation codes, and cyclic codes. For the selected fields of order at most $32$, the procedure produces longer MDS codes while preserving inequivalence to RS codes.
The resulting codes also exhibit different Schur-square behaviour. In particular, the extended codes with parameters $[9,4,6]$ over $\F_9$, $[9,4,6]$ over $\F_{13}$, and $[10,4,7]$ over $\F_{17}$ all satisfy $\dim(C^2)=9>2k=8$. Thus, the extension procedure can increase the Schur-square dimension beyond $2k$ while retaining both the MDS and non-RS properties. Over $\F_8$, the resulting $[10,3,8]$ codes attain the maximum possible length $q+2$ for dimension-three MDS codes.

Table~\ref{larger_fields} summarizes the extension sequences presented in this section. Here, $C_0$ denotes the initial code and $C_r$ denotes the code obtained after appending the listed $r$ columns.

\begin{table}[htbp]
    \centering
    \small
    \setlength{\tabcolsep}{5pt}
    \caption{Summary of the reported extensions of non-RS MDS codes }
    \label{larger_fields}
    \begin{tabular}{clcccc}
        \toprule
        Field & Source
        & \makecell{Initial\\code $C_0$}
        & \makecell{Extended\\code $C_r$}
        & $\dim(C_0^2)$ & $\dim(C_r^2)$\\
        \midrule
        $\F_{8}$  & Elliptic curve \cite{WLLZ}  & $[7,3,5]$ & $[10,3,8]$ & $6$ & $6$\\

        $\F_{9}$  & Elliptic curve  & $[7,4,4]$ & $[9,4,6]^*$ & $7$ & $9$\\

        $\F_{11}$ & TGRS \cite{GZ}   & $[8,3,6]$ & $[9,3,7]^*$ & $6$ & $6$\\

        $\F_{11}$ & TGRS \cite{GZ}   & $[8,4,5]$ & $[9,4,6]^*$ & $7$ & $8$\\

        $\F_{13}$ & Missing-monomial \cite{JMXZ}  & $[6,3,4]$ & $[10,3,8]^*$ & $6$ & $6$\\

        $\F_{13}$ & $(\ast)$-twisted \cite{BPR}   & $[7,3,5]$ & $[10,3,8]^*$ & $6$ & $6$\\

        $\F_{13}$ & CTRS \cite{LLWZ}  & $[7,4,4]$ & $[9,4,6]^*$ & $7$ & $9$\\

        $\F_{16}$ & $(+)$-twisted \cite{BPR}   & $[8,3,6]$ & $[10,3,8]$ & $6$ & $6$\\

        $\F_{17}$ & Cyclic \cite{XT}  & $[8,3,6]$ & $[12,3,10]^*$ & $6$ & $6$\\

        $\F_{17}$ & $(\ast)$-twisted \cite{BPR}   & $[9,4,6]$ & $[10,4,7]^*$ & $8$ & $9$\\

        \multirow{2}{*}{$\F_{17}$}
            & $(\ast)$-twisted \cite{BPR}   & $[9,6,4]$ & $[10,6,5]^*$ & $9$ & $10$\\

            & \emph{Duals of the row above} & $[9,3,7]$ & $[10,4,7]^*$ & $6$ & $10$\\

        $\F_{19}$ & $(\ast)$-twisted \cite{BPR}  & $[10,3,8]$ & $[14,3,12]^*$ & $6$ & $6$\\

        $\F_{19}$ & CTRS \cite{LLWZ} & $[10,5,6]$ & $[11,5,7]^*$ & $10$ & $10$\\

        $\F_{23}$ & $(\ast)$-twisted \cite{BPR}  & $[12,3,10]$ & $[14,3,12]^*$ & $6$ & $6$\\

        $\F_{27}$ & $(\ast)$-twisted \cite{BPR}  & $[14,3,12]$ & $[16,3,14]^*$ & $6$ & $6$\\

        $\F_{29}$ & Cyclic \cite{XT}   & $[7,3,5]$ & $[16,3,14]^*$ & $6$ & $6$\\

        $\F_{29}$ & $(\ast)$-twisted \cite{BPR}   & $[15,3,13]$ & $[16,3,14]^*$ & $6$ & $6$\\

        $\F_{32}$ & $(+)$-twisted \cite{BPR}   & $[17,3,15]$ & $[19,3,17]$ & $6$ & $6$\\
        \bottomrule
    \end{tabular}
    \small
    \parbox{\linewidth}{\centering Here $*$ indicates the code is not monomial equivalent to any RS code with the same parameters.}
\end{table}

\section{Conclusion and further study}

\qquad In this paper, we consider  extension of an MDS code. For a new column to be added, it should not lie in any hyperplane generated by $k-1$ columns of the current generator matrix. We first find the normal vectors of these hyperplanes. Then, we can check whether a new column is admissible by using inner products. In this way, the extension can be carried out step by step to obtain longer MDS codes. If the initial code is non-RS, the codes obtained by the extension are also non-RS.

The examples show that this procedure can be used for different codes over small finite fields. It can not only increase the length of the code, but also increase the dimension of its Schur square. Furthermore, some of the resulting codes have Schur squares of dimension strictly greater than $2k$, exceeding the dimension $2k$ attained by many existing constructions. These examples show that the procedure works in practice. Moreover, many new codes are non-RS and  have lengths larger than known non-RS MDS codes. In this sense, these codes are new. 

One limitation of the iterative procedure is that the choice of a new column at each step may affect the next extension. If we check all projective candidates and find that none of them can be added, then the current code cannot be extended by one more coordinate. However, this does not mean that the extension sequence we chose gives the longest code starting from the original code. In future work, we may use automorphisms, some algebraic properties, and backtracking to try different extension sequences and reduce the searching  space. It is also useful to compare the running time with other search methods to better understand the advantages and limitations of the proposed procedure.

\end{document}